\documentclass[11pt]{article}

\usepackage{amssymb}
\usepackage{graphicx}

 \usepackage{fullpage}

\usepackage{hyperref, subfigure}

\usepackage{epsfig}
\usepackage{algorithmic,algorithm}

\usepackage{amsthm}

\newcommand{\etal}{\emph{et~al.}}

\newtheorem{lemma}{Lemma}
\newtheorem{theorem}{Theorem}

\long\def\ignore#1{}

\begin{document}

%\mainmatter
\title{Conflict-free graph orientations\\ with parity constraints
%Work by T\'oth was partially supported by NSERC grant RGPIN 35586.
\thanks{Research was supported by NSF Grants \#CCF-0830734 and \#CBET-0941538. 
Work by T\'oth was also supported by NSERC grant RGPIN 35586.}
}
%\titlerunning{Conflict-free graph orientations with parity constraints}

%
\author{ Sarah Cannon
    \thanks{Department of Mathematics, Tufts University, Medford, MA 02155, USA, scanno01@cs.tufts.edu} 
    \thanks{Department of Computer Science, Tufts University,
      Medford, MA 02155, USA, mishaque@cs.tufts.edu}  
\and
Mashhood Ishaque\footnotemark[3]
\and
Csaba D. T\'oth \footnotemark[3] \thanks{Department of Mathematics and Statistics, University of Calgary, AB, Canada, cdtoth@ucalgary.ca} }

\date{}
%
%\authorrunning{Cannon, Ishaque, and T\'oth}
%\institute{Department of Mathematics, Tufts University, Medford, MA, USA
%\and
%Department of Computer Science, Tufts University, Medford, MA, USA
%\and
%Department of Mathematics and Statistics,
%University of Calgary, AB, Canada
%\vspace{2mm}
%\texttt{\\ scanno01@cs.tufts.edu, mishaque@cs.tufts.edu, cdtoth@ucalgary.ca}
%}

\maketitle

\vspace{-\baselineskip}

\begin{abstract}
It is known that every multigraph with an even number of edges has an even orientation ({\em i.e.}, all indegrees are even). We study parity constrained graph orientations under additional constraints. We consider two types of constraints for a multigraph $G=(V,E)$:
(1) an {\em exact conflict} constraint is an edge set $C\subseteq E$ and a vertex $v\in V$
such that $C$ should not equal the set of incoming edges at $v$;
(2) a {\em subset conflict} constraint is an edge set $C\subseteq E$ and a vertex $v\in V$ such that $C$
should not be a subset of incoming edges at $v$.
We show that it is NP-complete to decide whether $G$ has an even orientation
with exact or subset conflicts, for all conflict sets of size two or higher.
We present efficient algorithms for computing parity constrained orientations
with {\em disjoint} exact or subset conflict pairs.
\end{abstract}

\section{Introduction} \label{sec:intro}

%\vspace{-0.3\baselineskip}
An {\em orientation} of an undirected multigraph is an assignment of a direction to each edge.
It is well known~\cite{LP09} that a connected multigraph has an even orientation ({\em i.e.}, all indegrees are even) iff the total number of edges is even. In the {\em parity constrained orientation} ({\sc pco}) problem, we are given a multigraph $G=(V,E)$ and a function $p:V_0\rightarrow \{0,1\}$ for some subset $V_0\subseteq V$, and we wish to find an orientation of $G$ such that the indegree of every vertex $v\in V_0$ is $p(v)$ modulo 2, or report that no such orientation exists. This  problem has a simple solution in $O(|V| + |E|)$ time \cite{LP09}. Motivated by applications in geometric graph theory, we consider {\sc pco} under additional constraints of the following two types:

%\vspace{-0.7\baselineskip}
\begin{enumerate}
\item an {\em exact conflict} constraint is a pair $(C,v)\in 2^E\times V$ of a set $C$ of edges and a vertex $v$
     such that $C$ should not {\em equal} to the set of incoming edges at $v$;
\item a {\em subset conflict} constraint is a pair $(C,v)\in 2^E\times V$ of a set $C$ of edges and a vertex $v$
     such that $C$ should not be a {\em subset} of incoming edges at $v$;
\end{enumerate}

%\vspace{-0.3\baselineskip}

We denote by {\sc pco-ec} and {\sc pco-sc}, respectively, the {\sc pco} problem with exact conflicts and subset conflicts. We wish to find an orientation of $G$ such that the indegree of every vertex $v\in V_0$ is $p(v)$ modulo 2, and satisfies {\em all} additional constraints.

Two (exact or subset) conflicts, $(C_1,v_1)$ and $(C_2,v_2)$, are {\em disjoint} if $v_1\neq v_2$ or $C_1\cap C_2=\emptyset$. This means that disjoint conflicts at any fixed vertex correspond to disjoint edge sets. Let {\sc pco-dec} (resp., {\sc pco-dsc}) denote {\sc pco} with pairwise disjoint exact (resp., subset) conflicts.
If $|C|=k$ for some integer $k\in \mathbb{N}$ in all conflicts $(C,v)\in 2^E\times V$, we talk about the problems {\sc pco-$k$ec}, {\sc pco-$k$sc}, {\sc pco-$k$dec}, and {\sc pco-$k$dsc}. If $|C|=2$ for a conflict $(C,v)$, we say that $C$ is a {\em conflict pair} of edges.

\smallskip
\noindent {\bf Results.}
We show that {\sc pco-ec} and {\sc pco-sc} are NP-complete, and in fact already {\sc pco-2ec} and {\sc pco-2sc} are NP-complete. These problems are fixed parameter tractable: if $G$ has $m$ edges and there are $s_k$ conflicts of size $k=2,3,\ldots$, then they can be solved in $O( (m^{1.5} + n(n+m)) \prod_{k\geq 2}(k+1)^{s_k} )$ and $O( (n+m) \prod_{k\geq 2}k^{s_k})$ time, respectively. On the other hand, we present polynomial time algorithms for the variants with pairwise disjoint conflicts. Specifically, if the multigraph $G=(V,E)$ has $n$ vertices and $m$ edges, then both {\sc pco-dec} and {\sc pco-dsc} can be solved in $O(m^{2.5})$ time. For {\sc pco-2dec}, if no feasible orientation exists, we can compute an orientation with the {\em maximum} number of vertices satisfying the parity constraints within the same runtimes.

\smallskip
\noindent{\bf Motivation.} An even orientation subject to disjoint exact conflict pairs was a crucial tool in the recent solution of the {\em disjoint compatible matching conjecture}~\cite{ABD08,IST11} (see details below). The {\em exact conflict} constraint differs substantially from typical constraints in combinatorial optimization---it cannot be expressed as a linear inequality with 0-1 variables corresponding to the orientation of the edges. This led us to start a systematic study of {\sc pco-ec}. For comparison, we also considered {\em subset} conflicts, which have a natural integer programming representation. The two types of conflict constraints are indeed quite different.

\smallskip
\noindent{\bf Application.}
The exact conflict pair constraint originates from the disjoint compatible matching problem in geometric graph theory. It is clear that every 1-factor ({\em i.e.}, matching) can be augmented to a 2-factor by adding new edges. This, however, is not always possible if the input is a crossing-free straight-line graph in the plane, and it has to be augmented with compatible ({\em i.e.}, noncrossing) straight-line edges.

%%%%%%%%%%%%%%%%%%%%%%%%%%%%%%%%%%%%%%%%%%%%%%%%%%%%%%%%%%%%%%%%%%%%%%%%
%\vspace{-1.4\baselineskip}
\begin{figure}[hbtp]
  \centering
 \includegraphics[width=\columnwidth]{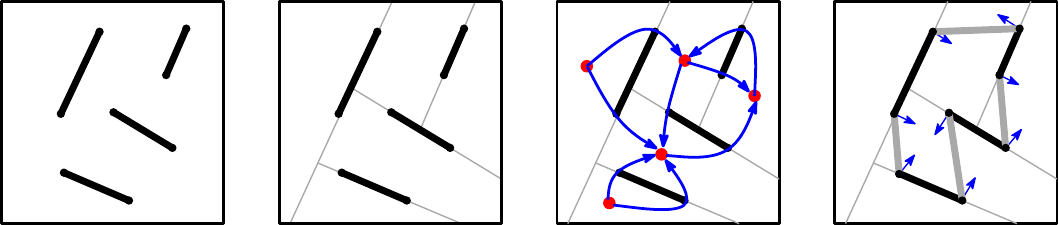}

  \vspace{-\baselineskip}
  \caption{\small An even geometric matching $M$. A convex decomposition.
  The dual graph with a conflict-free even orientation. An augmentation
  of $M$ to a 2-factor.\label{fig:geo-matchings}}
\end{figure}
%\vspace{-1.3\baselineskip}
%%%%%%%%%%%%%%%%%%%%%%%%%%%%%%%%%%%%%%%%%%%%%%%%%%%%%%%%%%%%%%%%%%%%%%%%%%%%%%%
It was conjectured~\cite{ABD08} that every geometric matching with an {\em even} number of edges
can be augmented to a crossing-free 2-factor (Disjoint Compatible Matching Conjecture). This conjecture
has recently been proved~\cite{IST11}. The new edges are added inside the faces of a convex
decomposition of the input matching. A crucial lemma in the proof claims that an augmentation to a 2-factor exists iff the {\em dual graph} of the convex decomposition has an even orientation that avoids a collection of pairwise disjoint exact conflict pairs~\cite{ABD08,IST11}. Our algorithm for {\sc pco-2dec} can decide whether such an orientation exists in $O(|M|^{2.5})$ time.

\smallskip
\noindent{\bf Related previous work.}
Graph orientations are fundamental in combinatorial optimization. It is a primitive often used for representing a variety of other problems. For example, {\em unique sink} orientations of polytopes are used for modeling pivot
rules in linear programming~\cite{SW01}, and {\em Pfaffian} orientations are used for
counting perfect matchings in a graph~\cite{LP09}. 

Hakimi~\cite{Hak65} gave equivalent combinatorial conditions for the existence of an edge orientation with prescribed indegrees. These were generalized by Frank~\cite{Fra80} for indegrees of subsets of vertices. Felsner~\etal~\cite{FFN10,FZ08} computed the asymptotic number of orientations with given indegrees.
The graph orientation problem where the indegree of each vertex must be between given upper and lower bounds was solved by Frank and Gy\'arf\'as~\cite{FG76}. Frank~{\em et al.}~\cite{FTS84} also solved the variant of this problem under parity constraints at the vertices. Frank~\etal~\cite{FJS01} characterized parity constrained graph orientations where the resulting digraph has $k$ edge-disjoint spanning arborescences with given roots. Frank and Kir\'aly~\cite{FK02} characterized graphs that admit $k$-edge-connected parity constrained orientations under any parity constraint where sum of parities equals the number of edges modulo~2.
Khanna~\etal~\cite{KNS05} devised approximation algorithms for an orientation constrained network design problem, but they do not consider parity or conflict constraints.

\smallskip
\noindent {\bf Proof techniques and organization.} The NP-hardness proofs and our algorithms are broken into elementary reduction steps, each of which uses some simple gadget, that is, a small graph with some carefully placed conflicts. These gadgets are quite remarkable and fun to work with, as they allow for a systematic treatment of all variants of the conflict-free graph orientation problem.

In Section~\ref{sec:npc}, we reduce {\sc (1-in-3)-sat} to {\sc pco-ec} and {\sc pco-sc}, independently. In Section~\ref{sec:alg}, we first reduce {\sc pco-2dec} to a maximum matching problem in a modified line graph. Then we reduce {\sc pco-dec}, with disjoint conflicts of size {\em at least 2,} to disjoint conflict {\em pairs}. Finally, the problem {\sc pco-dsc}, with disjoint subset conflicts, is reduced to {\sc pco-2dec}.

%\vspace{-\baselineskip}

\section{NP-completeness for exact and subset conflicts\label{sec:npc}}

%\vspace{-.5\baselineskip}

We reduce {\sc (1-in-3)-sat} to each of {\sc pco-2ec} and {\sc pco-2sc}. It follows that {\sc pco-$k$ec}, {\sc pco-$k$sc}, {\sc pco-ec} and {\sc pco-sc} are also NP-hard. {\sc (1-in-3)-sat} is known to be an NP-hard problem~\cite{AB09}. It asks whether a boolean expression in conjunctive normal form with 3 literals per clause can be satisfied such that each clause contains exactly one true literal.

%\vspace{-.5\baselineskip}

\subsection{NP-completeness of {\sc pco-2ec}}

%\vspace{-.5\baselineskip}

Let $I$ be an instance of {\sc (1-in-3)-sat} with variables $X_1,\ldots , X_n$ and clauses $C_1,\ldots , C_m$.
We construct a multigraph $G_I=(V,E)$ and a set $\mathcal{C}_I\subset {E\choose 2}\times V$ of exact conflict pairs (refer to Fig.~\ref{fig:reduction-ec}; arcs denote exact conflict pairs). For each variable $X_i$, we construct a caterpillar graph as follows. Create a chain of vertices labeled $x_{i1}, x_{i2}, x_{i3},\ldots , x_{i(2m)}$. Attach three edges to the first and the last vertex of this chain, and attach two edges to all interior vertices of the chain. We call these the {\em legs} of the caterpillar. At each vertex $x_{ij}$, let every pair of adjacent edges be an exact conflict pair.

%%%%%%%%%%%%%%%%%%%%%%%%%%%%%%%%%%%%%%%%%%%%%%%%%%%%%%%%%%%%%%%%%%%%%%%%
%\vspace{-1.3\baselineskip}
\begin{figure} %[hp]
  \centering
 \includegraphics[width=0.7\columnwidth]{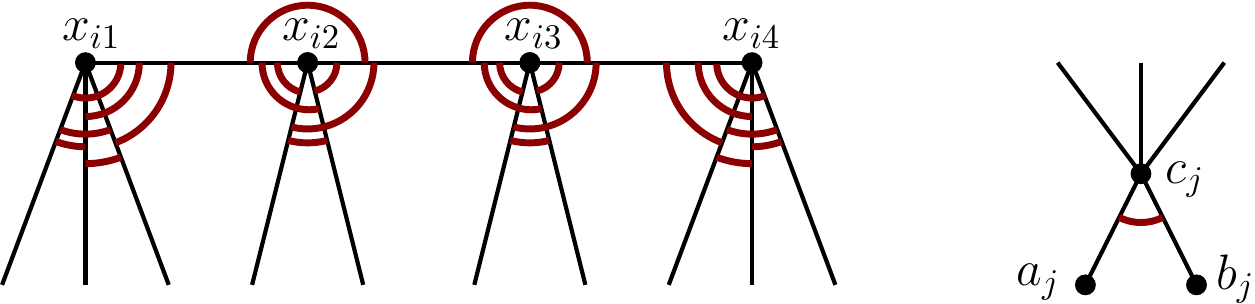}

 % \vspace{-\baselineskip}
  \caption{\small Left: variable gadget. Right: clause gadget. \label{fig:reduction-ec}}
\end{figure}
%\vspace{-1.3\baselineskip}
%%%%%%%%%%%%%%%%%%%%%%%%%%%%%%%%%%%%%%%%%%%%%%%%%%%%%%%%%%%%%%%%%%%%%%%%%%%%%%%

For each clause $C_j$, create a node $c_j$. We attach to $c_j$ a leg from each of the three caterpillars corresponding to the variables that appear in $C_j$. Specifically, if variable $X_i$ appears in clause $C_j$, attach some edge leaving vertex $x_{i(2j-1)}$ to $c_j$; if $\overline{X}_i$ appears in clause $C_j$, attach some edge leaving $x_{i(2j)}$ to $c_j$. At this point, each node $c_j$ has degree exactly 3, because each clause contains exactly three variables in instance $I$ of {\sc (1-in-3)-sat}. Additionally, for each node $c_j$ create two more nodes $a_j$ and $b_j$, each connected to $c_j$ by a single edge; make these edges an exact conflict pair.  Finally, create an additional node $v_0$, and connect all ``unused'' legs of the caterpillars to $v_0$. If $|E(G_I)|$ is odd, create an additional vertex $v_0'$ connected to $v_0$ by a single edge.
We will now show that a conflict-free even orientation of $G_I$ corresponds to a true instance $I$ of {\sc (1-in-3)-sat}.

\begin{lemma}
Instance $I$ of {\sc (1-in-3)-sat} is true iff there is an exact conflict-free even orientation for graph $G_I$ and set $\mathcal{C}_I$ of exact conflict pairs.
\end{lemma}

\begin{proof}
Assume instance $I$ of {\sc 1-in-3-sat} is true. That is, there is a valid truth assignment for all variables $X_i$ such that each clause $C_j$ contains exactly one true literal. We construct a conflict-free even orientation for $G_I$ and $\mathcal{C}_I$ as follows. If variable $X_i$ is true, then for all $\ell$, orient all edges incident to $x_{i(2\ell-1)}$ towards  $x_{i(2\ell-1)}$, and all edges incident to $x_{i(2\ell)}$ away from $x_{i(2\ell)}$. Note that the indegree of the vertices $x_{i\ell}$ in the caterpillar are alternately 0 and 4, and hence no exact conflict pair equals the set of edges oriented into one of these nodes. Similarly, if variable $X_i$ is false, then orient all edges incident to $x_{i(2\ell-1)}$ away from $x_{i(2\ell-1)}$, and all edges incident to $x_{i(2\ell)}$ towards $x_{i(2\ell)}$. Orient each edge $a_jc_j$ and $b_jc_j$ towards $c_j$. The legs of caterpillars oriented away from $c_j$ correspond to a `true'  assignment while legs oriented into $c_j$ correspond to a `false' assignment. Since exactly one literal in each $C_j$ is true, exactly one of 5 incident edges is oriented away from $c_j$. That is, the indegree of each $c_j$ is 4. We now have a conflict-free even orientation for $G_I$ and $\mathcal{C}_I$, as required.

Assume now that there is a conflict-free even orientation for $G_I$ and $\mathcal{C}_I$. The parity constraints and the conflict pairs ensure that all 4 edges incident to each $x_{i\ell}$ are oriented either to $x_{i\ell}$ or away from $x_{i\ell}$. Therefore, the indegrees of the nodes $x_{i1},\ldots , x_{i(2m)}$ are alternately 0 and 4. If all incident edges are oriented into $x_{i1}$, then set $X_i$ `true,' otherwise set $X_i$ `false.'
Our construction ensures that the indegree of each $c_j$ is exactly four. Since both $a_jc_j$ and $b_jc_j$ must be oriented into $c_j$, exactly two legs of some caterpillars are oriented to $c_j$ (and exactly one away from $c_j$). This guarantees that each $C_j$ contains exactly two false literals and one true literal, and so this truth assignment for all variables is a valid solution to instance $I$ of {\sc 1-in-3-sat}.
\end{proof}

By augmenting the conflict sets by additional edges, if necessary, we see that {\sc pco-$k$ec} is also NP-hard. It is clear that these problems are in NP: one can check in linear time whether the parity and all additional constraints are satisfied.
\begin{theorem}
Problems {\sc pco-ec} and {\sc pco-$k$ec}, for every $k\geq 2$, are NP-complete.
\end{theorem}

%\vspace{-1.3\baselineskip}

\subsection{NP-completeness of {\sc pco-2sc}}

%\vspace{-.5\baselineskip}

We now reduce {\sc (1-in-3)-sat} to {\sc pco-2sc}. Let $I$ be an instance of {\sc (1-in-3)-sat} with variables $X_1,\ldots , X_n$ and clauses $C_1,\ldots , C_m$. We construct a multigraph $G_I=(V,E)$ and a set $\mathcal{C}_I\subset {E\choose 2}\times V$ of subset conflict pairs (Fig.~\ref{fig:reduction-sc}; arcs denote subset conflict pairs).
For each variable $X_i$, create a circuit $(x_{i1},x_{i2},\ldots , x_{im})$ of length $m$. Label the edge connecting $x_{i\ell}$ and $x_{i,\ell+1}$ as $z_{il}$. (All arithmetic with index $\ell$ is performed modulo $m$). To each $x_{i\ell}$, attach two additional edges, $y_{i\ell}$ and $\overline{y}_{i\ell}$. Mark the edge pairs $\{z_{i\ell},z_{i,\ell+1}\}$, $\{z_{i\ell},y_{i,\ell+1}\}$, and $\{z_{i\ell},\overline{y}_{il}\}$ as subset conflict pairs.

For each clause $C_j$, create a node $c_j$. If $X_i$ is in clause $C_j$, attach the edge $y_{ij}$ to $c_j$. Similarly, if $\overline{X_i}$ occurs in clause $C_j$, attach the edge $\overline{y}_{ij}$ to $c_j$. At this point, each node $c_j$ should have degree 3, since each clause in instance $I$ of {\sc (1-in-3)-sat} contains three variables.
Label every pair of edges incident on $c_j$ as a subset conflict pair. Additionally, for each node $c_j$ create one more node $a_j$ connected to $c_j$ by a single edge. Finally, create an additional node $v_0$ and connect it to all ``unused'' edges $y_{i\ell}$ or $\overline{y}_{i\ell}$; if $|E(G_I)|$ is odd, create node $v_0'$ and connect it to node $v_0$ by a single edge.

%%%%%%%%%%%%%%%%%%%%%%%%%%%%%%%%%%%%%%%%%%%%%%%%%%%%%%%%%%%%%%%%%%%%%%%%
%\vspace{-1.3\baselineskip}
\begin{figure} %[htbp]
  \centering
 \includegraphics[width=.7\columnwidth]{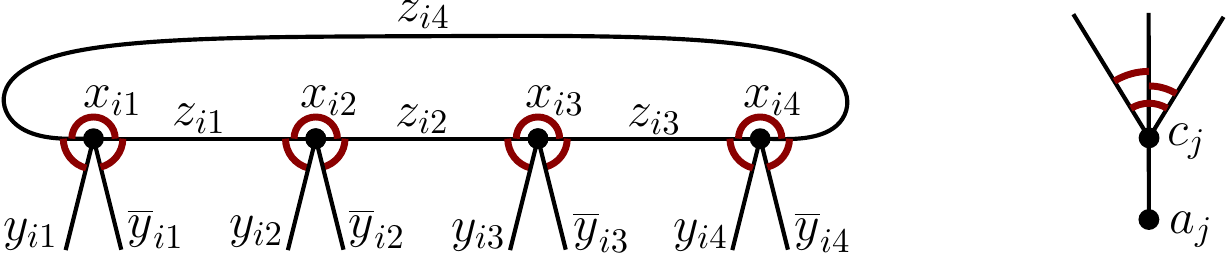}

  %\vspace{-\baselineskip}
  \caption{\small Left: variable gadget. Right: clause gadget.\label{fig:reduction-sc}}
\end{figure}
%\vspace{-1.3\baselineskip}
%%%%%%%%%%%%%%%%%%%%%%%%%%%%%%%%%%%%%%%%%%%%%%%%%%%%%%%%%%%%%%%%%%%%%%%%%%%%%%%

\begin{lemma}\label{lem:reduction-sc}
Instance $I$ of {\sc (1-in-3)-sat} is true iff there is a subset conflict-free even orientation for graph $G_I$ and set $\mathcal{C}_I$ of subset conflict pairs.
\end{lemma}

\begin{proof}
Assume that instance $I$ of {\sc (1-in-3)-sat} is true. That is, there exists a truth assignment for all variables $X_i$ such that exactly one literal in each clause $C_j$ is true. If $X_i$ is true, orient edge $z_{i\ell}$ from $x_{i\ell}$ to $x_{i,\ell+1}$ for all $\ell$; orient $y_{i\ell}$ away from $x_{i\ell}$; and orient $\overline{y}_{i\ell}$ into $x_{i\ell}$. Then, at each $x_{i\ell}$, the indegree is 2, but no two conflicting edges are oriented into $x_{i\ell}$. If $X_i$ is false, orient all edges of the variable gadget in the opposite way. Since each $C_j$ has exactly one true literal, exactly one of the three edges from variable gadgets is oriented into $c_j$. Orient edge $a_jc_j$ into $c_j$; now, the indegree of both $a_j$ and $c_j$ is even, and no two edges oriented into $c_j$ are in conflict.
We have constructed a conflict-free even orientation of $G_I$.

Assume now that there exists a conflict-free even orientation for $G_I$ and $\mathcal{C}_I$. The subset conflict pairs along the circuit $(x_{i1},\ldots , x_{im})$ ensure that the circuit is oriented cyclically. If $z_{i1}$ is oriented away from $x_{i1}$, then set $X_i$ to `true,' otherwise to `false.' In either case, exactly one edge of the circuit is oriented into each $x_{i\ell}$. Since the indegree has to be even, exactly one of $y_{i\ell}$ and $\overline{y}_{i\ell}$ is oriented into $x_{i\ell}$. The subset conflicts imply that if $X_i$ is true, $\overline{y}_{i\ell}$ is oriented into $x_{i\ell}$ and $y_{i\ell}$ is oriented away, while if $X_i$ is false, $y_{i\ell}$ is oriented into $x_{i\ell}$ and $\overline{y}_{i\ell}$ is oriented away.
In other words, edges oriented towards $c_j$ correspond to an assignment of `true' for the corresponding variable, while edges oriented away from $c_j$ correspond to assignments of `false' for the corresponding variables. For each node $c_j$, the conflicts we imposed ensure that exactly two edges are oriented into $c_j$,
and one of them is $a_jc_j$. Hence, exactly one variable in each clause has been set to true, as required.
\end{proof}

By augmenting the conflict sets with additional edges, if necessary, we see that {\sc pco-$k$sc} is also NP-hard. It is clear that these problems are in NP: one can check in linear time whether the parity and all additional constraints are satisfied.
\begin{theorem}
Problems {\sc pco-sc} and {\sc pco-$k$sc}, for every $k\geq 2$, are NP-complete.
\end{theorem}

%\vspace{-1.3\baselineskip}

\subsection{Fixed Parameter Tractability}

%\vspace{-.5\baselineskip}

We now show that {\sc pco-ec} and {\sc pco-sc} are fixed-parameter tractable.  If $G$ has $m$ edges and there are $s_k$ conflicts of size $k=2,3,\ldots$, then these problems can be solved in $O( (m^{1.5} + n(n+m)) \prod_{k\geq 2}(k+1)^{s_k} )$ and $O( (n+m) \prod_{k\geq 2}k^{s_k})$ time, respectively. 

First we consider {\sc pco-sc}. For each subset conflict set $S$ of size $k$ incident on vertex $v$, in any valid parity constrained orientation at least one edge of $S$ must be oriented away from $v$.   Arbitrarily choose one edge $e$ from each conflict set $S$ to be oriented away; there are at most $\prod_{k\geq 2}k^{s_k} = O(1)$ ways to do this.  Call this set of selected edges $E^{*}$.  For every edge $e$ in $E^{*}$, where $e$ is part of subset conflict set $S_e$, $e$ connects $v_e$, the vertex on which all edges in $S_e$ are incident, to $w_e$, some other vertex. To form a new graph $G'$, remove edge $e$ and if node $w_e$ has a parity constraint, reverse it.  A parity constrained graph orientation on $G'$ yields a solution to {\sc pco-sc} on $G$, obtained by reinserting all edges $e$ in $E^{*}$ with an orientation away from $v_e$ towards $w_e$.  If there is no parity constrained orientation on $G'$, repeat for some other set $E^{*}$.  If no set $E^{*}$ yields a parity constrained orientation on $G'$, there is no solution to {\sc pco-sc} on G.  If $m$ is the number of edges in $G$, which differs from the number of edges in $G'$ only by a constant, then it is known that each execution of the {\sc pco} algorithm takes $O(n+m)$ time \cite{LP09}.  The {\sc pco} algorithm is run at most $\prod_{k\geq 2}k^{s_k}$ times, giving a total polynomial runtime of $O( (n+m) \prod_{k\geq 2}k^{s_k})$.

Now we consider {\sc pco-ec}. We will reduce {\sc pco-ec} to {\sc pco-vd}, the parity constrained orientation problem in which each vertex has a minimum indegree given by a vertex demand function $F:V\rightarrow \mathbb{Z}$; this problem has already been solved by Frank~{\em et al.}~\cite{FTS84}. For each exact conflict set $S$ of size $k$ incident on vertex $v$, in any valid parity constrained orientation either (1) at least one edge must be oriented away, or (2) all edges are oriented toward $v$ and $v$ has indegree more than $k$.  For each exact conflict set, either choose an edge $e$ to orient away from $v$ or place a demand of $k+1$ on node $v$.  There are at most $\prod_{k\geq 2}(k+1)^{s_k} = O(1)$ ways to do this.  Call the set of selected edges $E^{*}$ and the set of selected vertices $V^*$. Remove the edges in $E^*$ as described previously, flipping the parity requirements of nodes $w_e$, to form graph $G'$. Search for a solution to {\sc pco-vd} on $G'$; such a solution yields a solution to {\sc pco-ec} when all edges $e$ in $E^*$ are reinserted and oriented away from $v_e$, towards $w_e$. Run for all possible sets $E^* \cup V^*$; if no valid solution to {\sc pco-vd} on $G'$ is found, there is no valid solution to {\sc pco-ec} on $G$.  If $n$ is the number of vertices and $m$ is the number of edges in $G$, then solving {\sc pco-vd}requires time $O(m^{1.5} + n(n+m))$ and
this algorithm runs in polynomial time at most $O( (m^{1.5} + n(n+m)) \prod_{k\geq 2}(k+1)^{s_k})$.

%\vspace{-.5\baselineskip}

\section{Polynomial time algorithms\label{sec:alg}}

%\vspace{-.5\baselineskip}

In this section we present polynomial time algorithms for {\sc pco-dec} and {\sc pco-dsc}.
We start by showing that in most cases we can restrict our attention to {\em even} orientations.
In the {\em even orientation} problem ({\sc eo}), we are given a multigraph $G=(V,E)$, and we wish to find an orientation of $G$ where {\em every} vertex has even indegree. Analogously to the variants of {\sc pco} with additional constraints, we introduce the problems {\sc eo-ec}, {\sc eo-dec}, and {\sc eo-$k$dec} for exact conflicts and {\sc eo-sc}, {\sc eo-dsc} and {\sc eo-$k$dsc} for subset conflicts.

We reduce {\sc pco} and most of its variants to the corresponding even orientation problems. The notable exception is {\sc pco-dec}.  We prove that the reduction holds for the more general optimization version of these problems as well. In the optimization version of {\sc pco} with possible additional constraints, we wish to find a conflict-free orientation which satisfies the maximum number of parity constraints.

%\vspace{-.5\baselineskip}

\begin{lemma}\label{lem:even}
The optimization versions of {\sc pco}, {\sc pco-ec}, {\sc pco-sc}, and {\sc pco-dsc}, can be reduced to the corresponding version of {\sc eo} in linear time.
\end{lemma}

\begin{proof}
We first reduce the parity constrained orientation ({\sc pco}) problem to the even orientation problem ({\sc eo}), and then consider various additional constraints. Consider an instance $I_1$ of {\sc pco}, that is, a multigraph $G_1=(V_1,E_1)$ with a partial parity constraint $p_1:V_0\rightarrow \{0,1\}$, $V_0\subseteq V_1$.
We construct an instance $I_3$ of {\sc eo} in two steps.

\noindent {\bf Step~1:} We construct an instance $I_2$ of {\sc pco} by augmenting $G_1$ to a multigraph
$G_2=(V_2,E_2)$ with new edges and vertices such that all parity constraints are even. %that is, $p_2:(V_0)_2\rightarrow \{0\}$ for some $(V_0)_2\subseteq V_2$.
For {\em each} vertex $v\in V_0$ with odd constraint $p_1(v)=1$, add a new (dummy) vertex $v'\in V_2$ and a new edge $vv'\in E_2$, with $p_2(v) = 0$ and $p_2(v') = 0$. Let $V'$ be the set of dummy vertices. If $G_1$ has an orientation satisfying $t_1$ out of $|V_0|$ parity constraints, then $G_2$ has an orientation satisfying $t_1+|V'|$ out of $|V_0|+|V'|$ parity ({\em i.e.}, evenness) constraints. Indeed, just orienting each dummy edge away from the dummy vertex means every dummy vertex has indegree 0, and the indegree of all adjacent vertices changes parity from odd to even. Conversely, if $G_2$ has an orientation satisfying the maximum number of parity constraints, say $t_2$, then we can assume that all dummy edges are oriented away from the dummy vertices. After deleting all dummy edges and vertices, we obtain an orientation of $G_1$ satisfying $t_2-|V'|$ parity constraints.

\noindent {\bf Step~2.} Consider an instance $I_2$ of {\sc pco}: a multigraph $G_2=(V_2,E_2)$ with even parity constraints $p_2:V_0\rightarrow \{0\}$ for some $V_0\subseteq V_2$. We construct a new instance of {\sc pco} in which the parity of {\em every} vertex is constrained to be even. Construct $G_3=(V_3,E_3)$ from $G_2$ by adding one new (dummy) vertex $w$, and connecting every vertex $v\in V_2\setminus V_0$ to $w$. If $|E_3|$ is odd, add one additional vertex $w'$ connected to $w$ by a single edge. Set the parity constraint of every vertex in $V_3$ to {\em even}. If $G_2$ has an orientation satisfying $t_2$ out of $|V_0|$ parity constraints, then $G_3$ has an orientation satisfying $t_2+|V_2\setminus V_0|+1$ parity constraints, just by orienting each dummy edge to make the parity of each unconstrained vertex even. Conversely, if $G_3$ has an orientation satisfying $t_3$ parity constraints, then after deleting the dummy vertices and edges (and also removing the parity constraints from vertices in $V_2\setminus V_0$) we obtain an orientation of $G_2$ satisfying $t_3-|V_2\setminus V_0|-1$ parity constraints.

In {\sc pco-ec}, an instance $I$ includes a family $\mathcal{C}$ of exact conflict constraints. We modify $\mathcal{C}$ as well in two steps. In the first step, we replace every conflict $(C,v)$ where $v\in V_0$, $|C|$ is {\em odd} and $p_1(v)=1$, with a new conflict $(C\cup \{vv'\},v)$. In the second step, we replace every conflict $(C,v)$ where $|C|$ is {\em odd} and $v\in V_2\setminus V_0$ with a new conflict $(C\cup \{vw\},v)$. These modifications ensure that after removing the dummy edges and vertices, the set of incoming edges are not in conflict at any vertex.

In {\sc pco-sc} and {\sc pco-dsc}, an instance $I$ includes a family $\mathcal{C}$ of subset conflict constraints. When we augment $G$ with new (dummy) vertices and edges, we preserve all these constraints. Independent of the orientation of the dummy edges, the constraints are satisfied in all feasible orientations for $I_1$, $I_2$, and $I_3$.
\end{proof}

%\smallskip\noindent
{\bf Remarks.}
With the above argument, every instance of {\sc pco-dec} can be reduced to an instance of {\sc eo-ec}, but the conflicts are no longer disjoint when we augment all conflicts at a vertex $v$ with a common dummy edge.

%\vspace{-1.5\baselineskip}

\subsection{Even orientations with disjoint exact conflict pairs\label{ssec:alg-2dec}}

%\vspace{-.3\baselineskip}

Let $G=(V,E)$ be a connected multigraph, and let $\mathcal{C}\subseteq {E\choose 2}\times V$ be a family of pairwise disjoint exact conflict pairs. We wish to find an orientation for $G$ with a maximum number of even vertices such that whenever a vertex $v$ has indegree 2 from edges $e_1$ and $e_2$, then $(\{e_1, e_2\},v) \notin \mathcal{C}$.
We present a polynomial time algorithm that either constructs an optimal orientation or reports that none exists. Without loss of generality, assume that $G$ is connected.

Recall the definition of the {\bf line graph} $L(G)$. Given a multigraph $G=(V,E)$, the nodes of $L(G)$ correspond to $E$, and two nodes are adjacent iff the corresponding edges of $G$ are adjacent. For a multigraph $G=(V,E)$ and conflict pairs $\mathcal{C}\subseteq {E\choose 2}\times V$, we define the following subgraph of $L(G)$. Let $L'=L'(G,\mathcal{C})$ be the graph whose node set is $E$, and two nodes $e_1,e_2\in E$ are adjacent in $L'$ iff they have a common endpoint $v\in V$ {\em and} $(\{e_1,e_2\},v)\notin \mathcal{C}$. We show that an instance of the
optimization version of {\sc eo-2dec} for $G$ and $\mathcal{C}$ reduces to a maximum matching over $L'(G,\mathcal{C})$.

\begin{lemma}\label{lem:linegraph}
Let $G=(V,E)$ be a multigraph with disjoint exact conflict pairs $\mathcal{C}$.
There are $t$ vertices with odd indegree in a conflict-free orientation of $G$ that
maximizes the number of even vertices iff there are $t$ nodes uncovered in
a maximum matching of $L'=L'(G,\mathcal{C})$.
\end{lemma}

\begin{proof}
First, suppose that a maximum matching $M$ of $L'$ covers all but $t$ nodes. We construct a conflict-free orientation for $G$. For every edge $(e_1,e_2)\in M$, direct both $e_1$ and $e_2$ towards one of their common endpoints in $G$. We obtain a partial orientation of $G$, where all indegrees are even, since pairs of edges are directed towards each vertex of $G$. Since adjacent but conflicting edges are not connected in $L'$, they are not matched in $M$, and thus there is no vertex in $G$ with indegree 2 where the two incoming edges are in conflict.

Now consider the set of unmatched nodes of $L'$, which is a set $E^*\subseteq E$ of edges in $G$ of size $|E^*|=t$. Out of any three edges incident to a common vertex, at least two can be matched, since the conflict pairs do not overlap. Hence each vertex $v\in V$ is incident to at most two edges in $E^*$; and if $v$ is incident to two edges in $E^*$, then those edges are in conflict. So the edges in $E^*$ form disjoint paths and circuits in $G$. We can orient the edges in $E^*$ into distinct vertices in $V$. We obtain an orientation of $G$ with exactly $t$ odd vertices.

Next suppose that in a conflict-free orientation of $G$ with the largest number of even vertices, there are exactly $t$ odd vertices. We construct a matching of $L'$. Consider a vertex $v\in V$. Partition the incoming edges of $v$ into two subsets whose size differ by at most one such that conflicting pairs are in different classes. This is possible, since the conflict pairs are disjoint, and so every edge participates in at most one conflict pair at $v$. Fix a maximum matching between the two classes arbitrarily. We have matched adjacent edges, but no conflicting pairs. If $v$ is even, then the matching is perfect, otherwise one edge remains unmatched. After repeating this for all vertices $v\in V$, we obtain a matching of $L'$ that covers all but $t$ edges in $E$.
\end{proof}

We use the following algorithm for constructing a desired even orientation. Given a multigraph $G$ and disjoint exact conflict pairs $\mathcal{C}$, construct graph $L'=L'(G,\mathcal{C})$, compute a maximum matching $M$ on $L'$, and convert it into an orientation of $G$. For a graph $G=(V,E)$ with $n$ vertices and $m$ edges, the line-graph $L'$ has $m$ nodes and $O(m^2)$ edges. The general max-flow algorithm used to find maximum matchings runs in time cubic in the number of nodes, or in $O(m^3)$ time. Since $L'$ is a unit-capacity graph, Dinic's blocking flow algorithm~\cite{AMO93} gives a runtime of $O(m^{2.5})$.

\vspace{-\baselineskip}

\subsection{Even orientations with disjoint exact conflicts}

We reduce {\sc pco-dec} with conflicts of size {\em at least} two to {\sc eo-2dec} in linear time. (Recall that {\sc pco-dec} has not been reduced to {\sc eo-dec} in Section~\ref{ssec:alg-2dec}). A key ingredient of the reduction is a ``switching network'' that can rearrange the orientations of $k$ edges of a conflict. This auxiliary network is defined as a graph $N_k$ with parity constraints and disjoint exact conflict pairs. It has $2k$ leaves: $k$ {\em input} leaves $a_1,\ldots ,a_k$ and $k$ {\em output} leaves $b_1,\ldots , b_k$. We draw $N_k$ in the plane such that the input leaves are on the left side, the output leaves are on the right side, and so it is convenient say that the orientation of each edge is either left-to-right (for short, {\em right}) or right-to-left ({\em left}). The network $N_k$ will have the following two properties:
\begin{itemize}\itemsep -1pt
\item[$\mathbf{P}_1$] If exactly $\ell$ input edges are oriented right, for some $0\leq \ell\leq k$, then exactly $\ell$ output edges are oriented right in {\em every} valid orientation of $N_k$.
\item[$\mathbf{P}_2$] If exactly $\ell$ input edges are oriented right, for some $0 < \ell < k$, then $b_1$ is oriented right and $b_2$ is oriented left in {\em some} valid orientation of $N_k$.
\end{itemize}

Properties $\mathbf{P}_1$ and $\mathbf{P_2}$ imply that outputs $b_1$ and $b_2$ represent all $k$ inputs for the purposes of exact conflicts. If all inputs are oriented right, then both $b_1$ and $b_2$ are oriented right; if no input is oriented right, then neither $b_1$ nor $b_2$ is oriented right. If some inputs are oriented right some are left, then there is a valid orientation where $b_1$ is oriented right and $b_2$ is left.

%%%%%%%%%%%%%%%%%%%%%%%%%%%%%%%%%%%%%%%%%%%%%%%%%%%%%%%%%%%%%%%%%%%%%%%%
%\vspace{-1.3\baselineskip}
\begin{figure}%[htbp]
  \centering
 \includegraphics[width=0.9\columnwidth]{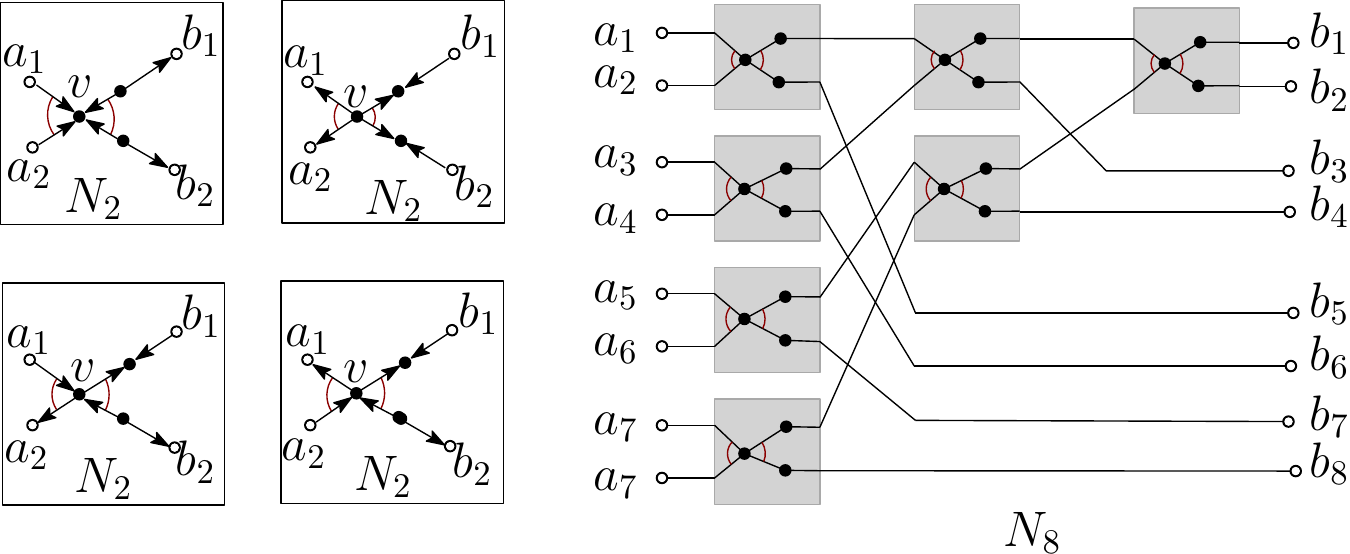}

 % \vspace{-\baselineskip}
  \caption{\small Left: $N_2$ with four possible orientations at $a_1$ and $a_2$. Right: $N_8$ is composed of 7 copies of $N_2$.\label{fig:switchnetowork}}
\end{figure}
%\vspace{-1.3\baselineskip}
%%%%%%%%%%%%%%%%%%%%%%%%%%%%%%%%%%%%%%%%%%%%%%%%%%%%%%%%%%%%%%%%%%%%%%%%%%%%%%%

For $k=2$, let $N_2$ be the graph shown in the left of Fig.~\ref{fig:switchnetowork} (arcs denote \emph{exact} conflict pairs). The leaves may have arbitrary indegrees, but every nonleaf vertex must have even indegree.

For every $k>2$, the graph $N_k$ is composed of multiple copies of $N_2$, similarly to a multi-stage switching network where the {\em switches} correspond to copies of $N_2$. Specifically, $N_k$ consists of $\lceil \log k\rceil$ stages. Stage $i=1,\ldots , \lceil \log k\rceil$ consists of $\lceil k/2^i\rceil$ copies of $N_2$. For each copy of $N_2$ at stage $i=1,2,\ldots \lceil \log k\rceil-1$, one output leaf is identified with an input leaf in the next stage, and the other output leaf becomes an output leaf of $N_k$. Refer to the right of Fig.~\ref{fig:switchnetowork} for an example with $k=8$. Note that graph $N_k$ has at most $6k$ nonleaf vertices.

\begin{lemma}
For every $k\in \mathbb{N}$, $k\geq 2$, graph $N_k$ satisfies both $\mathbf{P}_1$ and $\mathbf{P}_2$.
\end{lemma}
\begin{proof}
For $k=2$, the two disjoint conflict pairs at $v$ ensure that if both input edges are oriented right, then the indegree of $v$ is 4; if neither input edge is oriented right, then the indegree of $v$ is 0. If exactly one input edge is oriented right, then the indegree of $v$ is 2, and the second incoming edge may be any one of the two edges on the right side of $v$. It is now easy to verify that properties $\mathbf{P}_1$ and $\mathbf{P}_2$ hold.

For $k>2$, property $\mathbf{P_1}$ follows from the fact that $N_2$ has this property and we identified input edges with output edges in adjacent copies of $N_2$. For $\mathbf{P}_2$, assume that not all input edges have the same orientation. Consider an arbitrary valid orientation of $N_k$. If the two input edges of the rightmost copy of $N_2$ have different orientations, then $\mathbf{P}_2$ follows. Suppose that these two edges have the same orientation, say, both are oriented right. We show that $N_k$ has another valid orientation where these two edges have different orientations. Let $a_i$ be an input edge of $N_k$ oriented {\em left}. Note that $N_k$ contains a path from $a_i$ to the rightmost copy of $N_2$. In all copies of $N_2$ along this path, there is a valid orientation such that the edges between consecutive copies of $N_2$ are oriented left. Combining these orientations, we obtain a valid orientation where $b_1$ is oriented left and $b_2$ right, as required.
\end{proof}

Let $I$ be an instance of {\sc pco-dec} with conflicts of size {\em at least} 2. That is, $I$ consists of a multigraph graph $G=(V,E)$, a family of disjoint exact conflicts $\mathcal{C}\subseteq 2^E\times V$ each with at least two edges, and parity constraints $p: V_0\rightarrow \{0,1\}$. We may assume that at every vertex $v\in V_0$, the number of edges in each conflict set is $p(v)$ modulo~2, since all other conflict constraints are automatically satisfied. We create an instance $I'$ of {\sc eo-2dec}, that is, a multigraph $G' = (V',E')$ with disjoint conflict \emph{pairs} $\mathcal{C}'\subseteq {E'\choose 2}\times V'$ such that $G'$ has a conflict-free even orientation iff $G$ has a valid orientation.

%%%%%%%%%%%%%%%%%%%%%%%%%%%%%%%%%%%%%%%%%%%%%%%%%%%%%%%%%%%%%%%%%%%%%%%%
%\vspace{-1.5\baselineskip}
\begin{figure}%[htbp]
  \centering
 \includegraphics[width=0.7\columnwidth]{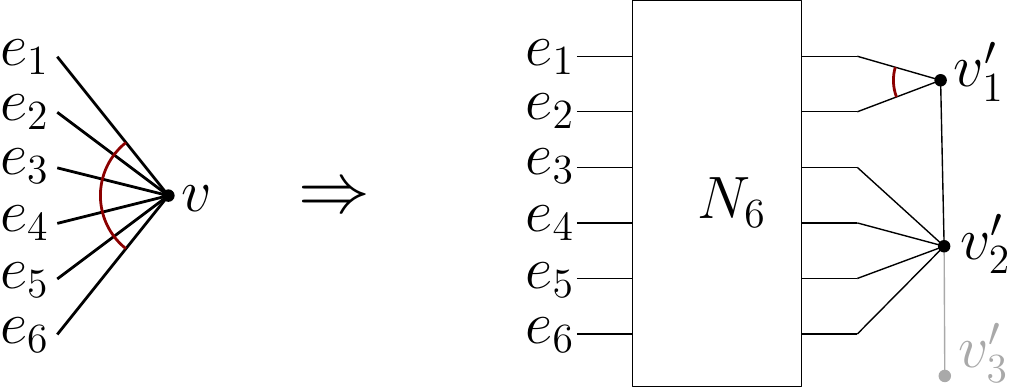}

  %\vspace{-\baselineskip}
  \caption{\small An exact conflict $(\{e_1,\ldots ,e_6\},v)$ is replaced by a network $N_6$ with the first two outputs identified with $v_1$ and all remaining outputs identified with $v_2$. \label{fig:reduction-dec}}
\end{figure}
%\vspace{-1.3\baselineskip}
%%%%%%%%%%%%%%%%%%%%%%%%%%%%%%%%%%%%%%%%%%%%%%%%%%%%%%%%%%%%%%%%%%%%%%%%%%%%%%%

For every vertex $v\in V$, we create a path $\pi'(v)$ in $G'$ as follows: If $p(v)=0$, then $\pi'(v)= (v_1',v_2',v_3')$ with three nodes; if $p(v)=1$ or $v\not\in V_0$, then $\pi'(v)= (v_1',v_2')$ with two nodes.
In order to balance the parity of unrestricted nodes $v\in V\setminus V_0$, we create one common auxiliary vertex $u_0'\in V$, and connect it to $v_2'$ for every $v\in V\setminus V_0$. If $|V\setminus V_0|$ is odd, we also add a dummy vertex $u_1'$ and a dummy edge $u_0'u_1'$ ($u_0'u_1'$ is oriented into $u_0'$ in any even orientation of $G'$).

For each edge $e\in E$, we create an edge $e'\in E'$ as follows. If $e$ is incident to $v$ and it is not part of any conflict at $v$, then let $e'$ be incident to $v_1$. For each conflict {\em pair} $(\{e_1,e_2\},v)$, let the corresponding edges, $e_1'$ and $e_2'$, be incident to $v_1$, and let $(\{e_1',e_2'\},v_1)\in \mathcal{C}'$ be an exact conflict pair. Finally, for each conflict $(\{e_1,\ldots ,e_k\},v)\in \mathcal{C}$, of size $k\geq 3$, we create a copy of the network $N_k$: identify the edges $e_1,\ldots e_k$ with the input leaves of $N_k$, identify output leaves $b_1$ and $b_2$ with $v_1$ forming an exact conflict pair at $v_1$, and identify the remaining $k-2$ output leaves with $v_2$. Fig.~\ref{fig:reduction-dec} shows an example for $k=6$. This completes the specification of the new instance $I'$ of {\sc eo-2dec}.

\begin{lemma}\label{lem:disjoint-exact}
Instance $I$ of {\sc pco-dec} with $G$,  $\mathcal{C}$, and parity constraints $p$ has a conflict-free orientation iff instance $I'$ of {\sc eo-2dec} with $G'$ and $\mathcal{C}'$ has a conflict-free even orientation.
\end{lemma}

\begin{proof}
Assume $G$ has a conflict-free parity constrained orientation $o$. We construct a conflict-free even orientation $o'$ for $G'$. Every edge $e\in E$ corresponds to an edge $e'\in E'$. We set the orientation of $e'$ to be the same as $e$. It remains to specify the orientation of auxiliary structures. For every vertex $v\in V$, we orient edge $v_1'v_2'\in E'$ to make the indegree of $v_1'$ even; and then the possible edges $v_2' v_3'$ or $v_2' u_0'$ are oriented to make the indegree of $v_2'$ even. Since $o'$ satisfies the parity constraints at every vertex $v\in V$,
and we added an dummy edge $v_3'v_2'$ oriented into $v_2$, it follows that the indegrees of all vertices $v_1'$, $v_2'$, and (if exists) $v_3'$ are even. Next, we choose the orientations in the networks $N_k$. For a conflict set $(\{e_1,\ldots , e_k\},v)\in \mathcal{C}$, a network $N_k$ forwards two edges to a conflict pair at $v_1$ and the remaining $k-2$ edges to $v_3$. By properties $\mathbf{P}_1$ and $\mathbf{P}_2$, the conflict pair has 0 (resp., 1 or 2) edges oriented into $v_1$ iff $\{e_1,\ldots , e_k\}$ has 0 (resp., $0 < \ell<k$ or $k$) edges oriented into $v$. This implies that if $o$ has no conflict at $v\in V$, then $o'$ has no conflict at $v_1'\in V'$.

Assume now that $G'$ has a conflict-free even orientation $o'$. We construct a conflict-free parity-constrained orientation $o$ on $G$. Recall that every edge $e\in E$ corresponds to an edge $e'\in E'$. Let each $e$ take the same orientation in $o$ as $e'$ has in $o'$. Suppose that the set of incoming edges at a vertex $v\in V$ equals a conflict set $\{e_1,\ldots , e_k\}$. Then the set of incoming edges of $v_1'$ is the conflict pair $\{e_1',e_2'\}$, that is, $o'$ is not a conflict-free orientation. It follows that $o$ is a conflict-free orientation.
\end{proof}

%\vspace{-1.5\baselineskip}

\subsection{Even orientations with disjoint subset conflicts}

%\vspace{-.5\baselineskip}

We reduce {\sc eo-dsc} to {\sc eo-2dec} in linear time. Let $I$ be an instance of {\sc eo-dsc}, that is, a multigraph $G=(V,E)$ with disjoint subset conflicts $\mathcal{C}$ of various sizes. We construct a new multigraph $G' = (V',E')$ with disjoint exact conflict {\em pairs} such that $G$ has a conflict-free even orientation iff $G'$ does.

%%%%%%%%%%%%%%%%%%%%%%%%%%%%%%%%%%%%%%%%%%%%%%%%%%%
\begin{figure}%[hptb]
  \centering
 \includegraphics[width=0.9\columnwidth]{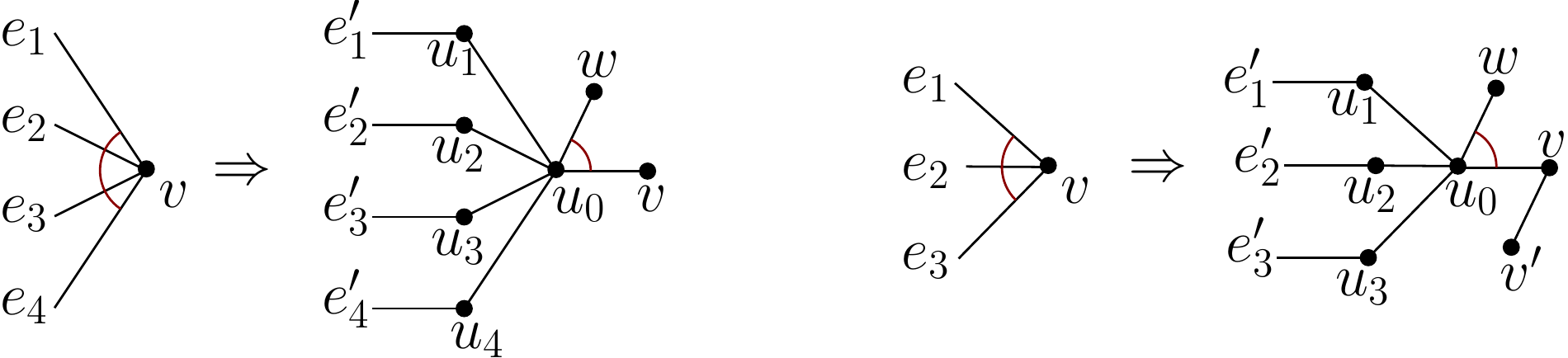}

  \caption{\small Left: Modification for a subset conflict $(\{e_1,e_2,e_3,e_4\},v)$ of even size.
  Right: Modification for a subset conflict $(\{e_1,e_2,e_3\},v)$ of odd size.
  \label{fig:reduction-dsc}}
\end{figure}
%%%%%%%%%%%%%%%%%%%%%%%%%%%%%%%%%%%%%%%%%%%%%%%%%%%%

The graph $G'$ is constructed by modifying $G=(V,E)$. We make some local modifications for each subset conflict $(\{e_1,\ldots ,e_k\},v)\in \mathcal{C}$.
If $k$ is even, then replace the edges $e_1,\ldots ,e_k$ with the configuration shown in Fig.~\ref{fig:reduction-dsc}(left) with $k+2$ new vertices $u_0,u_1,\ldots , u_k,w$ and one new exact conflict pair $(\{u_0v,u_0w\},u_0)\in \mathcal{C}'$. If $k$ is even, then replace the edges $e_1,\ldots ,e_k$ with the configuration shown in Fig.~\ref{fig:reduction-dsc}(right) with $k+3$ new vertices  $u_0,u_1,\ldots , u_k,v',w$ and one new exact conflict pair $(\{u_0v,u_0w\},u_0)\in \mathcal{C}'$. By construction,
the new exact conflict pairs in $\mathcal{C}'$ are pairwise disjoint.

\begin{lemma}\label{lem:subset}
Instance $I$ of {\sc eo-dsc} with $G$ and subset conflict $\mathcal{C}$ has a conflict-free even orientation iff instance $I'$ of {\sc eo-2dec} with $G'$ and $\mathcal{C}'$ has a conflict-free even orientation.
\end{lemma}

\begin{proof}
Assume $G$ has a conflict-free even orientation $o$. We construct a conflict-free even orientation $o'$ for $G'$.  The common edges of $G$ and $G'$ should have the same orientation as in $o$. For each conflict $(\{e_1,\ldots , e_k\},v)\in \mathcal{C}$, let the orientation of the edges $e_1',\ldots ,e_k'\in E'$ be the same as $e_1,\ldots , e_k\in E$, respectively. Since $o$ is conflict-free, some edge $e_i\in E$, $1\leq i\leq k$, is oriented away from $v$. The corresponding edge $e_i'\in E'$ is oriented away from $u_i$, and so $u_0u_i$ is oriented into $u_0$.
Therefore, $\{u_0v,u_0w\}$ cannot be the set of edges oriented into $u_0$. Since the indegrees of $u_0,u_1,\ldots , u_k$ are even (0 or 2), they uniquely determine the orientation of {\em all} edges incident to $u_0$. Note also that the edge $wu_0$ is always oriented into $v$ because the indegree of $w$ must be 0. It follows that $u_0v$ is oriented into $v$ in $o'$ iff an odd number of edges in $\{e_1,\ldots , e_k\}$ are oriented into $v$ in $o$. That is, the contribution of a conflict to the parity of $v$ is the same as the contribution of the edge $u_0v$ in $o'$. Overall, if $o$ is an even orientation, then $o'$ is even as well.

Assume now that $G'$ has a conflict-free even orientation $o'$. We construct a conflict-free even orientation $o$ for $G$. The common edges of $G$ and $G'$ should have the same orientation as in $o'$. For each conflict $(\{e_1,\ldots , e_k\},v)\in \mathcal{C}$, let the orientation of the edges $e_1,\ldots ,e_k\in E$ be the same as $e_1',\ldots , e_k'\in E'$, respectively. Since $o'$ is conflict-free and the indegree of $u_0$ is even, some edge $u_0u_i$, $1\leq i\leq k$, is oriented into $u_0$. This means that $e_i'\in E'$ is oriented away from $u_i$, and the corresponding edge $e_i\in E$ is not oriented into $v$. Hence, not all edges in $\{e_1,\ldots , e_k\}$ are orieted into $v$. All vertices of $G$ preserve their parity, hence $o$ is even as well.
\end{proof}

%\vspace{-\baselineskip}

%\noindent {\bf Running time.}
%For the graph $G$ with $n$ vertices and $m$ edges, $N$ has $O(m + n)$ nodes and $O(m)$ links.
%It can be made unit-capacity by introducing $\sum_{v}{f(v)} = O(m)$ auxiliary nodes, thus using
%Dinic's blocking flows algorithm we can find the maximum flow in $O(m \sqrt{m+n}))$ time.
%To find augmenting paths, we perform at most $n$ breadth-first searches in $O(n(m+n))$ time.
%Thus the total time complexity of the algorithm is $O(m^{1.5}+n(n+m))$.

%\vspace{-1.5\baselineskip}

\section{Conclusion\label{sec:con}}

%\vspace{-.5\baselineskip}

We have shown that the parity constrained orientation problem is NP-hard in the presence of exact or subset conflicts, and in fact already in the presence of conflict pairs. On the other hand, the problems are in P for disjoint conflict pairs. It remains an open problem to  determine the status of {\sc PCO-DEC} if all conflicts have one or two edges; while subset conflict sets with one edge are trivial, exact conflict sets with one edge are not, and our reductions only apply to exact conflicts with two or more edges.  
%\vspace{-.5\baselineskip}


\begin{thebibliography}{1}

%\vspace{-.5\baselineskip}
\bibitem{AMO93}
R.~K. Ahuja, T.~L. Magnanti, and J.~B. Orlin, {\em Network Flows: Theory,
  Algorithms, and Applications}, {Prentice Hall}, February 1993.

\bibitem{ABD08}
O.~Aichholzer, S.~Bereg, A.~Dumitrescu, A.~Garc\'ia, C.~Huemer, F.~Hurtado,
M.~Kano, A.~M\'arquez, D.~Rappaport, S.~Smorodinsky, D.~Souvaine, J.~Urrutia,
and D.~R. Wood,
Compatible geometric matchings,
{\em Comput. Geom.} {\bf 42} (2009), 617-–626.

%\bibitem{AHH09}
%O.~Aichholzer, T.~Hackl, M.~Hoffmann, A.~Pilz, G.~Rote, B.~Speckmann, and B.~Vogtenhuber.
%Plane graphs with parity constraints,
%{\em Proc.~11th WADS}, vol.~5664 of LNCS, Springer, 2009, pp.~13--24.

\bibitem{AB09}
S.~Arora and B.~Barak,
{\em Computational Complexity: A Modern Approach},
Cambridge Univ.~Press, 2009.

%\bibitem{BBCS11}
%M.~Bra\v{c}i\v{c}, D.~Bokal, \'E.~Czabarka, and L.~Sz\'ekely,
%Orientations give lower bounds on crossing numbers,
%presentation at {\em ``Crossing Numbers Turn Useful''}, BIRS, 2011.
%equivalent to nding certain optimal orientations in an appropriately defined multigraph.
%Consequences of this results include the fact that a planar graph always has an orientation with every
%indegree being at most 3 (Fraysseix, Ossona de Mendez), and every graph has an orientation where the
%maximum indegree is at most the ceiling of twice the maximum average degree, where the maximum is
%taken over all subgraphs (Aichholzer-Aurenhammer-Rote, Venkateshvaran).

%\bibitem{DPSW11}
%A.~Darmann, U.~Pferschy, J.~Schauer, and G.~J.~Woeginger,
%Paths, trees, and matchings under disjunctive constraints,
%{\em Discrete Appl. Math.} {\bf 159} (16) (2011), 1726--1735.

\bibitem{FFN10}
S.~Felsner, \'E.~Fusy, and M.~Noy,
Asymptotic enumeration of orientations,
{\em Discrete Math. Theor. Comp. Sci.} {\bf 12} (2) (2010), 249--262.

\bibitem{FZ08}
S.~Felsner and F,~Zickfeld,
On the number of planar orientations with prescribed degrees,
{\em Electron. J. Comb.} {\bf 15} (1) (2008), article R77.

\bibitem{Fra80}
A.~Frank, On the orientaiton of graphs,
{\em J. Combin. Theor. B} {\bf 28} (1980), 251--261.

\bibitem{FG76}
A.~Frank and A.~Gy\'arf\'as,
How to orient the edges of a graph,
{\em Coll. Math. Soc. J. Bolyai} {\bf 18} (1976), 353–-364.

\bibitem{FJS01}
A.~Frank, T.~Jord\'an, and Z.~Szigeti,
An orientation theorem with parity conditions,
{\em Discrete Appl. Math.} {\bf 115} (2001), 37--47.

\bibitem{FK02}
A.~Frank and Z.~Kir\'aly,
Graph orientations with edge-connection and parity constraints,
{\em Combinatorica} {\bf 22} (2002), 47–-70.

\bibitem{FTS84}
A.~Frank, \'E.~Tardos, and A.~Seb\H{o}, Covering directed and odd cuts,
{\em Math Prog. Stud.} {\bf 22} (1984), 99--112.

\bibitem{Hak65}
S.~L.~Hakimi, On the degrees of the vertices of a directed graph,
{\em J. Franklin Inst.} {\bf 279} (1965), 280--308.

\bibitem{IST11}
M.~Ishaque, D.~L.~Souvaine, and C.~D.~T\'oth,
Disjoint compatible geometric matchings,
in {\em Proc. 27th Sympos. on Comput. Geom.}, ACM, 2011, pp.~125--134.

\bibitem{KNS05}
S.~Khanna, J.~Naor, and F.~B.~Shepherd,
Directed network design with orientation constraints,
{\em SIAM J. Discre. Math.} {\bf 19} (2005), 245--257.

\bibitem{LP09}
L.~Lov\'asz and M.~D. Plummer, {\em Matching Theory},
AMS Chelsea, 2009.

\bibitem{SW01}
T.~Szab\'o and E.~Welzl, Unique sink orientations of cubes, in {\em Porc.~42nd FOCS}, IEEE, 2001, pp.~547--555.

\end{thebibliography}
\end{document}